\documentclass{fac}

\usepackage{graphicx}
\usepackage{amssymb}
\usepackage{amsfonts}
\usepackage{amsmath}
\usepackage{dsfont}
\usepackage{MnSymbol}
\usepackage{mathtools}
\usepackage{epsfig}
\usepackage{epstopdf}
\usepackage{tikz}
\usepackage{amsthm}
\ifprodtf \else \usepackage{latexsym}\fi

\newtheorem{theorem}{Theorem}[section]

\title[Draft of Relation of Web Service Orchestration, Abstract Process, Web Service and Choreography]
      {Relation of Web Service Orchestration, Abstract Process, Web Service and Choreography}

\author[Yong Wang]
    {Yong Wang\\
     College of Computer Science and Technology,\\
     Faculty of Information Technology,\\
     Beijing University of Technology, Beijing, China\\
     }

\correspond{Yong Wang, Pingleyuan 100, Chaoyang District, Beijing, China.
            e-mail: wangy@bjut.edu.cn}

\pubyear{2017}
\pagerange{\pageref{firstpage}--\pageref{lastpage}}

\begin{document}
\label{firstpage}

\makecorrespond

\maketitle

\begin{abstract}
We refine the relation of Web service orchestration, abstract process, Web service, and Web service choreography in Web service composition, under the situation of cross-organizational corporation. We also introduce the formal verification process of this relation through an example.
\end{abstract}


\section{Introduction}\label{int}

Web service composition is a core mechanism in Web service research domain, it can compose component Web services (WSs) into bigger granularity of WSs. There are two aspects of Web service composition: one is Web service orchestration (WSO) which usually has an abstract process (AB), the other is Web service choreography (WSC). In a framework of formal model of Web service composition \cite{FWSC}, we designed an architecture of Web service composition under the situation of business to business (B2B) corporation. In this paper, we refine this architecture to clarify the relation of WSO, AB, WS, and WSC, and introduce the formal verification process of this relation by use of truly concurrent process algebra APTC \cite{APTC} through an example.

This paper is organized as follows. We do not introduce the preliminaries on truly concurrent process algebra APTC and Web service composition, please refer to \cite{APTC} and \cite{FWSC} for details. In section \ref{rel}, we refine the relation of WSO, AB, WS, and WSC. We introduce the formal verification process of the relation by APTC through an example in section \ref{ver}. Finally, in section \ref{con}, we conclude this paper.

\section{Relation of WSO, AB, WS and WSC}\label{rel}

The relation among WSO, AB, WS and WSC are illustrated in Fig. \ref{relation}. It was firstly occurred in a formal model of Web service composition \cite{FWSC}, and we detail this relation here.

\begin{enumerate}
  \item WSO: A WSO is an orchestration of workflow activities, including internal activities, activities interacting with internal applications, activities interacting with WSs outside via its interface WS. WSO locates in the interior of organization and usually is executed with in WSO engine;
  \item AB: An abstract process is an abstraction of a WSO by eliminating the inner activities and activities interacting with inner applications, it is only a description of  the interface of a WSO with virtual control flow and data flow, and it can not usually be executed. An AB corresponds to its interface WS exactly with each virtual activity corresponding to a WS operation in its interface WS;
  \item WS: A WS is the real interface of a WSO interacting with other WSOs or WSs outside, and it is located in the DMZ and may be equipped within a Web (Service) server;
  \item WSC: A WSC is the specification of interactions between WSs, it may have an entity or just act as a contract of involved WSs.
\end{enumerate}

\begin{figure}
  \centering
  \includegraphics{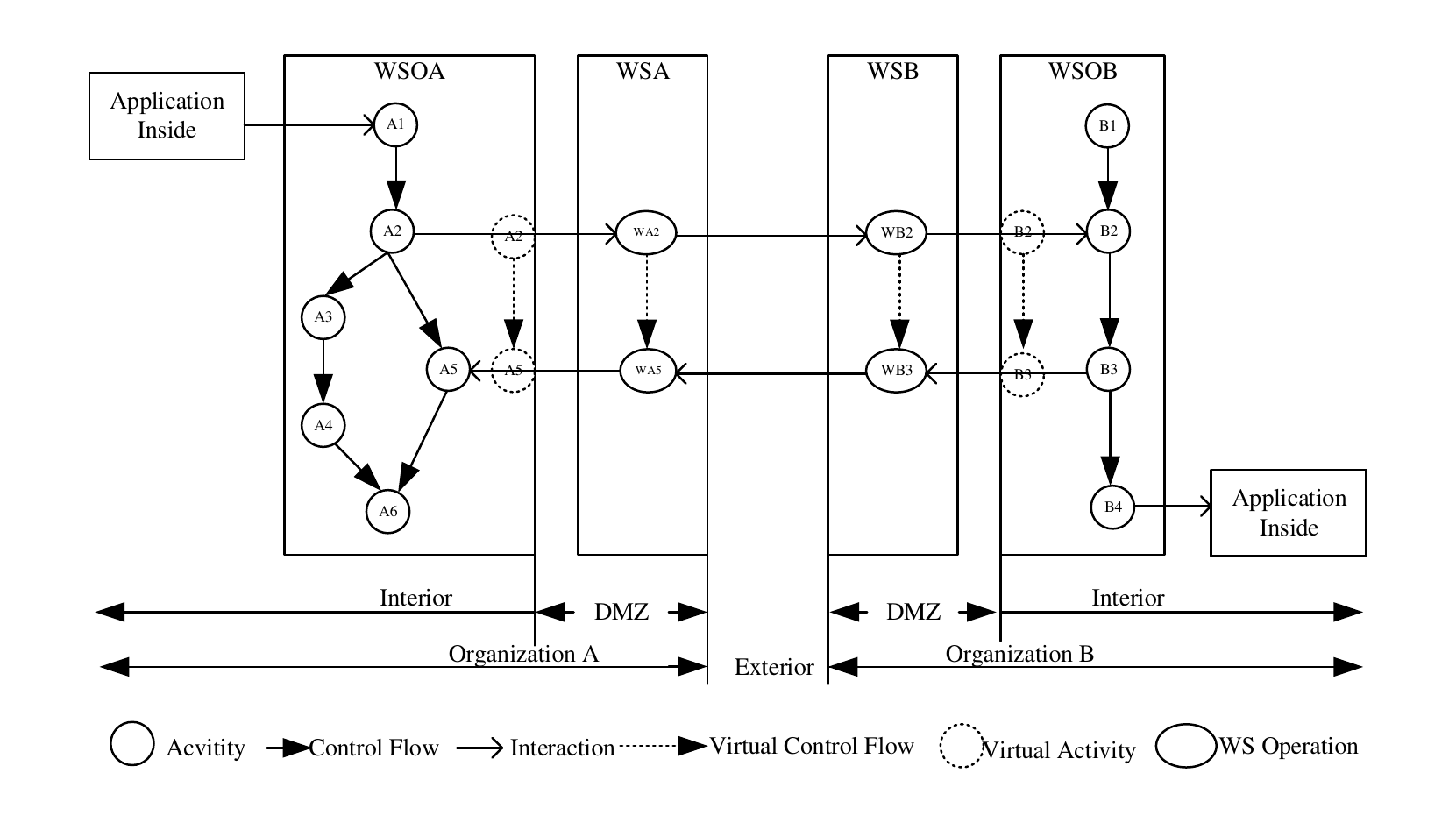}
  \caption{Relation of WSO, AB, WS and WSC.}
  \label{relation}
\end{figure}

\section{Verification of the Relation}\label{ver}

The above relation can be formally verified by truly concurrent process algebra APTC \cite{APTC}. The verification process is following.

\begin{enumerate}
  \item Giving the APTC descriptions of WSOs, WSs;
  \item By abstracting the internal activities, activities interacting with internal applications, the ABs can be obtained;
  \item Putting WSOs, WSs in parallel, treating interactions between WSO and WS, WS and WS, WS and WSO as internal actions, and by use of recursion, we can verify the whole system if is correct.
\end{enumerate}

We take the example illustrated in Fig. \ref{relation}, and verify its correctness.

The APTC description of WSOA is following.

$WSOA=\sum_{d\in \Delta}A1(d)\cdot WSOA1$

$WSOA1=A2\cdot WSOA2$

$WSOA2=((A3\cdot A4)\parallel A5)\cdot WSOA3$

$WSOA3=A6\cdot WSOA$

By defining $I=\{A1(d),A3, A4, A6\}$, we can get $ABA=\tau_I(WSOA)=A2\cdot A5$.

The APTC description of WSA is as follows. Note that, to make the interacting action in the same causal depth, we use the shadow constant.

$WSA=\circledS_{A1}\cdot WSA1$

$WSA1=WA2\cdot WSA2$

$WSA2=WA5\cdot WSA3$

$WSA3=\circledS_{A6}WSA$

The APTC description of WSOB is following.

$WSOB=B1\cdot WSOB1$

$WSOB1=B2\cdot WSOB2$

$WSOB2=B3\cdot WSOB3$

$WSOB3=\sum_{d'\in\Delta}B4(d')\cdot WSOB$

By defining $I=\{B1,B4\}$, we can get $ABB=B2\cdot B4$.

The APTC description of WSB is as follows. Note that we also use the shadow constant for the same reason.

$WSB=\circledS_{B1}\cdot WSB1$

$WSB1=WB2\cdot WSB2$

$WSB2=WB3\cdot WSB3$

$WSB3=\circledS_{B4}\cdot WSB$

We define the following communication function, and the communications between any other two atomic actions will cause deadlock.

$\gamma(A2, WA2)\triangleq c(A2, WA2)$

$\gamma(A5, WA5)\triangleq c(A5, WA5)$

$\gamma(B2, WB2)\triangleq c(B2, WB2)$

$\gamma(B3, WB3)\triangleq c(B3, WB3)$

$\gamma(WA2, WB2)\triangleq c(WA2, WB2)$

$\gamma(WA5, WB3)\triangleq c(WA5, WB3)$

Then, we define $I=\{A2, A3, A4, A5, A6, \circledS_{A6}, WA2, WA5, WB2, WB3, B1, \circledS_{B1}, B2, B3, \\ c(A2, WA2), c(A5, WA5), c(B2, WB2), c(B3, WB3), c(WA2, WB2), c(WA5, WB3)\}$

We can get the system $\tau_I(\partial_H(\Theta(WSOA\between WSA\between WSB\between WSOB)))$, and the following result.

\begin{theorem}
The system $\tau_I(\partial_H(\Theta(WSOA\between WSA\between WSB\between WSOB)))$ has desired external behaviors.
\end{theorem}

\begin{proof}
By use of the axioms of APTC, deducting on the term $\tau_I(\partial_H(\Theta(WSOA\between WSA\between WSB\between WSOB)))$, we can get

$\tau_I(\partial_H(\Theta(WSOA\between WSA\between WSB\between WSOB)))\\=\sum_{d\in\Delta}\sum_{d'\in\Delta}A1(d)\cdot B4(d')\cdot \tau_I(\partial_H(\Theta(WSOA\between WSA\between WSB\between WSOB)))$
\end{proof}

\section{Conclusions}\label{con}

We give the refined relation among WSO, AB, WS and WSC based on a framework of Web service composition \cite{FWSC}, and introduce the formal verification process of the relation. Through an example illustrated in Fig. \ref{relation}, we detail the verification process.

\newpage

%

\label{lastpage}

\end{document}